\documentclass[sigplan,nonacm]{acmart}

\usepackage{amsmath}
\usepackage{booktabs}
\usepackage{array}
\usepackage{multirow}
\usepackage{graphicx}
\usepackage{xcolor}
\usepackage{xspace}
\usepackage{algorithm}
\usepackage{algpseudocode}

\hypersetup{colorlinks=true,linkcolor=blue,citecolor=blue,urlcolor=blue}

\newcommand{\sys}{SemBridge\xspace}
\newcommand{\cset}{C^{\mathrm{req}}}
\newcommand{\obs}{\mathcal{O}}
\newcolumntype{P}[1]{>{\raggedright\arraybackslash}p{#1}}
\newtheorem{proposition}{Proposition}

\begin{document}

\title[SemBridge]{SemBridge: Compiling Consumer Observations into Cross-Stack Communication Plans}

\author{Genlang Chen}
\email{cgl@zju.edu.cn}
\affiliation{%
  \institution{NingboTech University}
  \city{Ningbo}
  \country{China}
}

\author{Junyi Zhu}
\email{JunyiZhu.cs@gmail.com}
\affiliation{%
  \institution{Dalian Ocean University}
  \city{Dalian}
  \country{China}
}

\author{Yuanshan Lin}
\email{linyuanshan@dlou.edu.cn}
\affiliation{%
  \institution{Dalian Ocean University}
  \city{Dalian}
  \country{China}
}

\renewcommand{\shortauthors}{G. Chen, J. Zhu, and Y. Lin}

\begin{abstract}
Distributed-tensor systems specify where values reside, while collective systems optimize how requested operations execute. At a boundary between vendor runtimes that cannot share a native communicator, neither abstraction states what a remote consumer must observe. \sys fills this gap by compiling graph and runtime facts into a typed contract for the consumer-visible result and its delivery obligations. The contract captures provenance, substitutability, completion, authority, demand, and native-domain locality. A deterministic lowerer constructs backend-neutral communication plans, and a symbolic checker validates each plan before execution across CUDA/NCCL and CANN/HCCL. An independent layout-only planner handles all 72 structural transitions but establishes only 54 complete obligations; a byte-only minimizer proposes 40 semantically invalid candidates, all rejected by \sys. On nine real edges, \sys produces distinct observation-aware plans that reduce startups on all nine and payload bytes on the three result edges. A live CUDA/CANN run derives and executes full-logit reconstruction, source projection, and owner-token delivery from log-probability, token-only, and owner-scoped requests. On a measured two-host 1-GbE capacity-spillover deployment, source projection cuts result traffic by more than 99.97\% and increases throughput by 8.92--80.20\% across Dense, MoE, and MiniMax workloads. All 18 MiniMax restart pairs at concurrency 1, 8, and 16 favor source projection. A Qwen3-14B MLP slice additionally verifies bitwise activation-shard delivery and HCCL completion of row-parallel partials. These results establish consumer observation as a semantic layer between placement and collective execution.
\end{abstract}

\maketitle

\section{Introduction}

Large language models increasingly exceed the capacity of one accelerator pool, while available resources can belong to distinct hardware and software stacks. Heterogeneous inference systems respond by selecting model placements, parallelism strategies, and request routes over devices with different compute, memory, and network characteristics~\cite{mei2025helix,jiang2024hexgen,jiang2025hexgen2}. Phase-disaggregated serving further transfers request state between specialized prefill and decode pools~\cite{patel2024splitwise,zhong2024distserve}. These systems make heterogeneous placement practical. Across CUDA/NCCL and CANN/HCCL, however, endpoints remain in disjoint native communicator domains. Once a partition is fixed across them, what must cross the boundary for remote consumers to observe the correct value?

Distributed-tensor compilers describe layouts, propagate sharding, and generate parallel programs~\cite{shazeer2018mesh,xu2021gspmd,alabed2025partir}. Placement systems search over inter- and intra-operator parallelism~\cite{jia2019flexflow,zheng2022alpa,jia2022whale}, and cross-mesh resharding optimizes transfers between layouts~\cite{zhuang2023optimizing}. Collective systems synthesize or tune a physical algorithm for a requested AllReduce, AllGather, broadcast, or related primitive~\cite{cai2021sccl,shah2023taccl,xu2025autoccl,hei2026heteccl}. Mixed-vendor substrates can realize such collectives through cross-domain point-to-point transport and vendor-local combining operations~\cite{wang2026hetccl}. \sys augments these layout and collective abstractions with the result a remote consumer observes: a projection, one substitutable replica, a completed value, or an owner-issued decision. It lowers that requirement before physical collective selection.

Autoregressive result transfer makes the gap concrete. In our integration, the native tail output head materializes full logits. A layout-directed plan slices the tensor by TP lane, transfers the shards, reconstructs the full value with AllGather, and then samples. A token-only request can apply the same deterministic sampler before the boundary and communicate \texttt{int32} token IDs; a log-probability request still requires complete logits. Selected homogeneous serving paths already communicate compact sampler outputs~\cite{vllm2026ppoutput,vllmascend2026sampling}. \sys generalizes this path-specific choice into a derived cross-stack transformation: the serving API determines the consumer surface, the graph and runtime determine its delivery obligations, and the checker validates the resulting plan. The same contract composes projection with lane-preserving shards, partial completion, substitutable replicas, authoritative decisions, and native-domain constraints. Thus identical endpoints can admit different legal payloads and communication graphs.

\begin{figure*}[t]
  \centering
  \includegraphics[width=\textwidth]{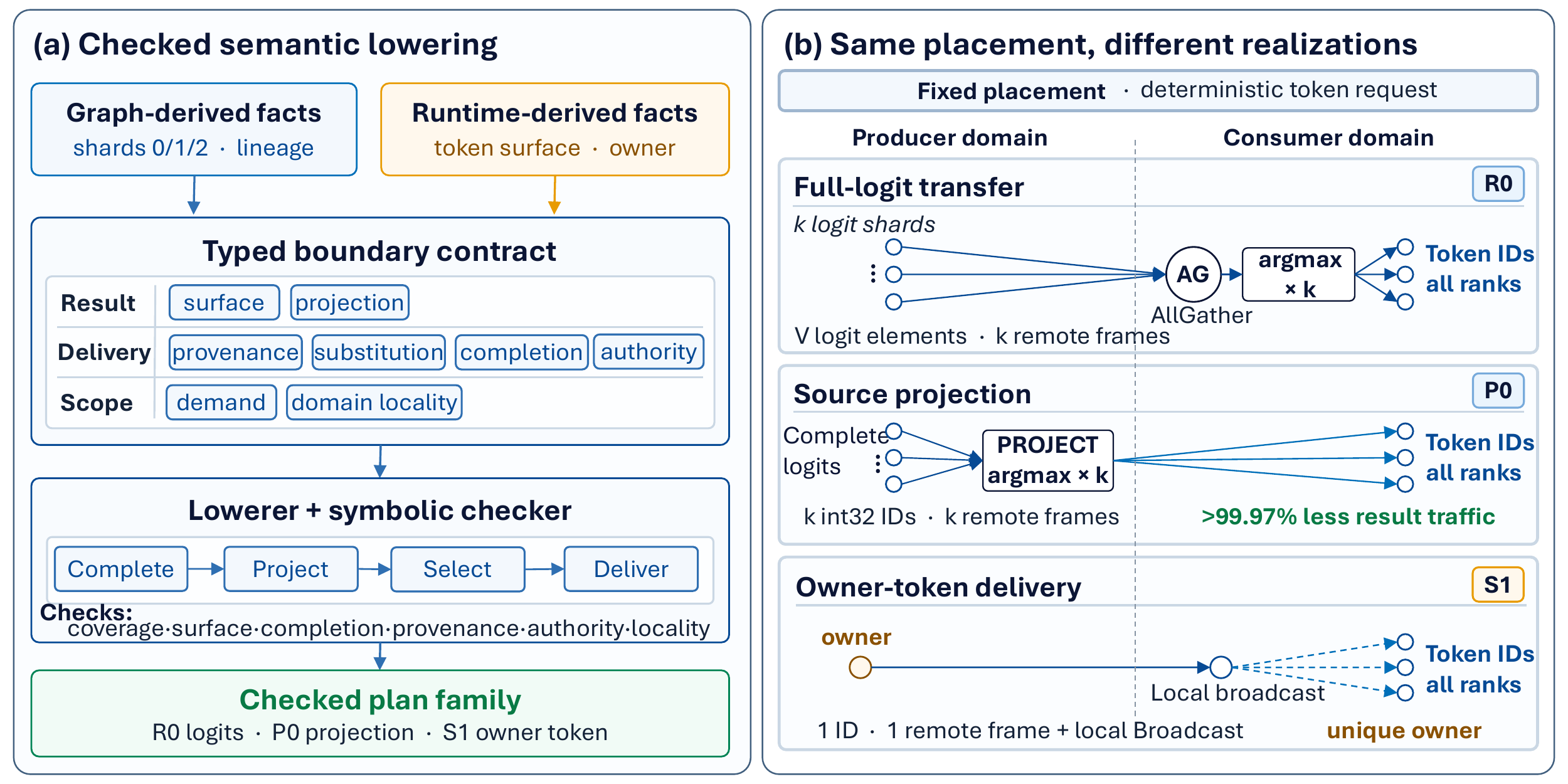}
  \caption{Consumer-observation-guided lowering. (a) Graph and runtime facts form a typed boundary contract; an ordered lowerer and symbolic checker produce a checked plan. (b) The dashed line separates producer and consumer domains under one fixed placement. Full-logit transfer (R0), source projection (P0), and owner-token delivery (S1) preserve the same requested token ID. R0 transfers $V$ logit elements in $k$ remote frames, P0 transfers $k$ \texttt{int32} IDs in $k$ frames, and S1 transfers one ID in one remote frame followed by local Broadcast.}
  \Description{Two-panel overview of graph and runtime facts flowing through a typed contract, ordered lowering, and symbolic checking, followed by three communication plans that preserve the same requested token ID with different cross-domain payload and frame counts.}
  \label{fig:motivation}
\end{figure*}

Figure~\ref{fig:motivation} separates two decisions: which result the consumer observes and how that result must be delivered. For a token-only request, full-logit transfer and source-side projection expose the same token surface; they differ only in where argmax runs. Log-probability requests retain complete-logit transfer. Delivery rules separately preserve shard mappings, complete partials, select representatives from derived replica classes, and source decisions from their owner. SBP and DTensor encode shard, replicate, and partial states~\cite{yuan2021oneflow,pytorch2026dtensor}; \sys adds the consumer surface, demand, substitutability, authority, execution identity, and native-domain facts.

An independent planner isolates what layout information alone can establish. It receives shape, dtype, meshes, domains, and Shard/\allowbreak Replicate/\allowbreak Partial placements. Although it supports all 72 layout transitions, it establishes only 54 complete application obligations. On those 54 cases, \sys reduces bytes in 30 and startups in 35. Authority adds another 10 byte and 13 startup reductions across the remaining 18 cases. On nine real edges, \sys changes every plan, reducing startups on all nine and bytes on the three result edges. A shared-IR oracle matches the production lowerer throughout the bounded family; a byte-only minimizer proposes 40 semantically invalid candidates, all rejected by \sys.

All plans run through the same persistent transport and vendor-local collectives. In the projection-placement comparison, forward traffic and result-frame counts match. Source-side projection reduces result traffic by 99.979--99.996\% and improves throughput by 8.92\% on Dense TP4 and 26.39--62.87\% on MoE TP2. On MiniMax, it improves throughput by 18.64--80.20\% across concurrency 1, 8, and 16, with all 18 restart pairs positive. Separate controls show that fewer forward startups help Dense, while owner-token fanout helps MoE when result synchronization lies on the critical path.

This work makes three contributions:
\begin{itemize}
  \item A typed boundary contract separates the consumer-visible result from delivery obligations over provenance, substitution, completion, and authority, while retaining demand and native-domain locality.
  \item A deterministic lowerer compiles the contract into a backend-neutral graph. A symbolic checker validates coverage, surface, completeness, provenance, authority, and domain locality before runtime realization.
  \item A plan-parametric CUDA/\allowbreak NCCL--CANN/\allowbreak HCCL runtime executes API-selected full-logit, projected-token, and owner-token plans together with representative, exact-shard, and partial-completion plans. Experiments cover three model families and MiniMax concurrency from 1 to 16.
\end{itemize}

\begin{table*}[t]
  \centering
  \caption{Abstraction boundary relative to the closest systems.}
  \label{tab:capability}
  \small
  \begin{tabular}{P{0.22\textwidth} P{0.34\textwidth} P{0.36\textwidth}}
    \toprule
    Research thread & Primary representation or optimization & Relation to \sys \\
    \midrule
    OneFlow SBP; PyTorch DTensor~\cite{yuan2021oneflow,pytorch2026dtensor} & Split/shard, broadcast/replicate, partial, and layout redistribution & Supplies tensor state; \sys adds consumer result surfaces, demand, authority, execution identity, and cross-domain observation checks \\
    GSPMD; PartIR; cross-mesh resharding~\cite{xu2021gspmd,alabed2025partir,zhuang2023optimizing} & Sharding propagation, partition rewrites, and layout-directed multicast & Selects or realizes placement transitions; \sys selects the consumer-observable requirement for a fixed transition \\
    CoCoNet; Unity~\cite{jangda2022coconet,unger2022unity} & Joint computation/communication programs and verified graph substitutions & Rewrites program structure; \sys checks per-boundary observations and native-domain locality \\
    P$^2$; TACCL; MSCCLang~\cite{xie2022reduction,shah2023taccl,cowan2023mscclang} & Reduction programs and executable collective schedules & Implements a specified reduction or collective after \sys selects what information must cross \\
    HetCCL~\cite{wang2026hetccl} & Mixed-vendor device transport and hierarchical collectives & Realizes a requested mixed-vendor collective; \sys produces the checked logical requirement supplied to such a backend \\
    \bottomrule
  \end{tabular}
\end{table*}

\section{Cross-Domain Communication Semantics}

\subsection{Execution Model}

Let $\mathcal{D}$ be the set of native collective domains. A domain contains endpoints that can participate directly in one vendor runtime's communicator and collective semantics. For a boundary edge $e$, let $P_e$ be its producer endpoints, $C_e$ its possible consumer endpoints, and $\cset_e\subseteq C_e$ the consumers that demand the value. The function $D(x)\in\mathcal{D}$ maps endpoint $x$ to its native domain. A cross-domain operation connects endpoints whose domain identifiers differ; it is expressed in a backend-neutral logical graph before a transport implementation is chosen.

\sys accepts model placement and parallelism from a human plan, a distributed-tensor compiler, or a heterogeneous placement system~\cite{zheng2022alpa,alabed2025partir,mei2025helix}. It occupies the interface between placement and physical collective execution, where it identifies the weakest logical communication graph that preserves the observation required by each consumer.

\subsection{Placement and Consumer Observation}

Tensor layouts remain necessary. An exact-shard mapping identifies which producer supplies each consumer; a resharding operation changes the mapping between device meshes; a partial value records unfinished reduction state. Distributed-tensor systems already represent and propagate these properties~\cite{shazeer2018mesh,xu2021gspmd,zhuang2023optimizing,alabed2025partir}. The application boundary also contains facts that placement does not express: which consumers are active, which result surface they observe, whether replica substitution is valid, whether completion remains outstanding, and which producer owns a request-scoped decision.

The distinction changes both payload and graph shape. Lane $i$ must retain its producer mapping when it carries a distinct shard. A graph-derived replica may use one representative per remote domain. A partial must be reduced before a consumer that expects the complete value. A decision must originate at its owner. At a result edge, complete logits require shard transfer and reconstruction, whereas a token-only consumer permits deterministic sampling at the producer and transfer of the projected token. The endpoint layout can remain fixed in every case.

\subsection{Consumer Surfaces and Delivery Obligations}

We factor consumer observation into two orthogonal dimensions, $\textit{result surface}\times\textit{delivery obligations}$. The surface names what the consumer API exposes: complete logits, log probabilities, token IDs, or another registered projection. Four delivery obligations constrain how that result arrives. A \emph{provenance-preserving obligation} requires the mapped producer's value. A \emph{substitution obligation} permits a representative when graph lineage places the producers in one downstream-substitutable class. A \emph{completion obligation} reduces an unfinished partial before a complete-value observation. An \emph{authority-preserving obligation} permits only the runtime-selected owner to source a decision or persistent state. Deterministic token-only sampling derives the token surface by applying argmax to complete logits; lowering chooses whether that projection runs before or after the boundary. Complete-value, log-probability, and stochastic-sampling requests retain complete logits because the current contract does not carry RNG state. Projection is a result-surface transformation orthogonal to the four delivery classes evaluated in the normalized matrix.

The result surface and delivery obligation first determine the legal plan set. A fabric-aware selector can then choose among accepted plans using message multiplicity, local collective cost, load, and the execution critical path. Sections~\ref{sec:lowering} and~\ref{sec:evaluation} separate this semantic filtering from runtime selection.

\subsection{Scope of the Current Backend}

The serving adapter realizes replicated forward values, reconstructs complete logits, projects tokens, and returns owner-sourced tokens. A companion Qwen3-14B MLP path uses the same production transport for lane-preserving exact shards and HCCL completion of row-parallel partials. Together, these paths exercise complete-value, projected-result, provenance, substitution, completion, and authority observations on CUDA/CANN engines.

\section{Consumer-Observable Contracts}

\subsection{Boundary Contract}

For each boundary edge, \sys constructs the compile-time contract
\begin{equation}
e=\langle P,\cset,D,K,\Gamma,R,O,Q\rangle .
\label{eq:contract}
\end{equation}
$P$ is the producer set, $\cset$ is the demanded consumer set, and $D$ is the native-domain mapping. $K$ selects the value kind. $\Gamma$ partitions producers into downstream-substitutable classes derived from graph lineage. $R$ records an outstanding reduction and the required completion, $O\subseteq P$ is the authoritative owner set, and $Q$ names the consumer-visible result surface and a registered projection operator, including its input/output types and deterministic tie-breaking semantics. Shape, dtype, tensor layout, lifetime, and framing group accompany the contract.

Autoregressive execution introduces a separate runtime tag
\begin{equation}
\tau=\langle\textit{request},\textit{token},\textit{microbatch},\textit{epoch}\rangle .
\label{eq:tag}
\end{equation}
The contract answers which logical value a consumer may observe; $\tau$ identifies the execution instance to which a delivered frame belongs. Keeping the two concepts separate prevents request lifecycle metadata from becoming part of the core communication abstraction.

The frontend compiles graph structure and runtime APIs into this representation. Tensor-parallel AllReduce lineage derives complete replica classes; partial lineage derives the reduction obligation; the request graph derives demanded consumers; and runtime APIs identify result outputs and the decision owner. The three evaluated model partitions use the same derivation path and no model- or configuration-specific lowering branch.

\subsection{Consumer Observation}

The semantic value of an execution is defined at demanded consumers rather than at every intermediate rank. We use observational refinement at a registered communication boundary, with equivalence defined over the demanded observations named by its contract. For consumer $c$, the observation is
\begin{equation}
\obs_e(c)=\langle\textit{surface},\textit{value},\textit{complete},\textit{provenance},\textit{authority}\rangle .
\label{eq:observation}
\end{equation}
A runtime delivery is current only when its received tag matches the expected $\tau$. A plan $p$ is legal if
\begin{equation}
\forall c\in\cset_e,\quad
\obs_p(e,c)\equiv\obs_{\mathrm{ref}}(e,c),
\label{eq:legality}
\end{equation}
where the reference follows the complete layout-directed execution for the requested consumer surface. For a token-only request, both executions are compared after the registered projection; for log probabilities, the surface remains complete logits plus token IDs. Authority additionally requires the decision to originate at its runtime-selected owner.

Equation~\ref{eq:legality} defines observational refinement: intermediate state may change or disappear while demanded consumers retain the same registered result with admissible completeness, provenance, and authority. A representative plan need not reproduce all producer-to-consumer copies, and a projection plan need not reconstruct state absent from the consumer surface.

\subsection{Legality Invariants}

The checker symbolically executes each candidate operation. It initializes producer values as available, tracks surface, completeness, and provenance through transfer and projection, and updates availability at destinations. Reduction combines provenance and discharges completion. Acceptance requires coverage of every demanded consumer, the requested surface, complete values where required, admissible producer provenance, substitution only within a derived class, discharge of outstanding reductions, an authorized source for every decision, and native collectives confined to one domain.

The checker rejects unavailable sources, missing consumers, cross-domain operations mislabeled as local, unfinished reductions, invalid provenance, and non-owner decision sources. The runtime separately rejects stale request, token, and state-epoch tags at receive boundaries. This division keeps compile-time semantic soundness independent of the details of framed transport.

\subsection{Semantic Inputs}

Optimization-authorizing fields come from the graph or runtime API. A completed tensor-parallel AllReduce establishes a replica class, partial lineage names its reduction, the request graph supplies active consumers, and the serving API supplies requested outputs, deterministic-sampling mode, and the decision owner. Runtime byte-equality gates validate evaluated executions but do not authorize lowering. If a required fact is unavailable, the frontend emits the corresponding layout-materializing plan; log-probability requests therefore retain complete-logits redistribution.

\subsection{Contract Construction and Static Validation}

A boundary contract is built at the integration point between model partitioning and transport. Endpoint groups identify producer and consumer lanes and their native collective domains. Graph placement supplies lane-preserving, replicated, or partial relations; request/runtime metadata supplies demand, result surface, sampling mode, and owner selection. Shape and dtype determine both the materialized footprint and the projected result type.

Reduction obligations name an elementwise operator and source set. The executable reference supports sum, product, minimum, and maximum; the live Qwen3 path uses HCCL sum AllReduce. Lowering completes the registered operator before exposing a full value.

Static validation checks endpoint membership, domain assignment, mapped sources, owner uniqueness, request identity, reduction lineage, observation compatibility, and projection preconditions before lowering.

The frontend then materializes an immutable edge object with field-level provenance. Graph- and runtime-derived facts enable specialized rules, validation observations remain non-authorizing evidence, and missing facts select the reference obligation. Placement chooses endpoints and layouts; the contract compiler derives the consumer observation; and \sys verifies that the communication graph preserves both.

\subsection{Soundness Argument}

\begin{proposition}[Observation-preserving refinement]
Assume valid graph/runtime facts, correct implementations of the registered projection semantics, correct native-domain collectives, and reliable tagged delivery. Every communication plan accepted by \sys preserves the conservative observation of every demanded consumer.
\end{proposition}

\noindent\textit{Proof sketch.}
Initially, every producer or authoritative owner has one available value with its own provenance and surface. The checker permits an operation to read only available sources. A point-to-point transfer preserves value, completeness, and provenance. A representative transfer requires compatible derived classes. A reduction consumes the registered producers, combines provenance, and marks the result complete. A projection applies the operator named by $Q$ to complete producer state, preserving its registered input, output, and tie-breaking semantics. Ownership checks restrict decision sources, domain checks confine native collectives, and coverage makes every $c\in\cset$ available. Induction over the operation sequence establishes Equation~\ref{eq:legality}; runtime tag matching selects the execution instance. \hfill$\square$

Section~\ref{sec:evaluation} tests these invariants through differential reference executions, targeted invalid mutations, and live result-observation checks.

\section{Checked Semantic Lowering}
\label{sec:lowering}

\subsection{Lowering and Checking Procedure}

Algorithm~\ref{alg:lowering} formalizes the compiler pipeline. It separates completion, result-surface transformation, source selection, and cross-domain delivery into ordered passes. The lowerer emits backend-neutral operations, and the checker enforces coverage, surface, completion, provenance, authority, and native-domain locality before backend execution. The complete rule matrix appears in the Supplemental Material.

\begin{algorithm}[t]
  \caption{Contract-guided lowering and legality checking}
  \label{alg:lowering}
  \small
  \begin{algorithmic}[1]
    \Require boundary edge $b$, graph facts $G$, runtime facts $A$
    \Ensure checked logical plan $p$, or a rejected contract
    \State $e \gets \Call{BuildContract}{b,G,A}$
    \State $\Call{ValidateContract}{e}$
    \State $p \gets [\,]$; $q \gets [\,]$; $v \gets \Call{ProducerState}{e}$
    \If{$e$ has an outstanding completion obligation}
      \State $(p,v) \gets \Call{Complete}{p,v,e.R}$
    \EndIf
    \If{$e.Q$ is token-only and $\Call{CanProject}{e}$}
      \State $(p,v) \gets \Call{ProjectAtProducer}{p,v,e.Q}$
    \ElsIf{$e.Q$ includes a consumer projection}
      \State $(p,v) \gets \Call{MaterializeComplete}{p,v,e}$
      \State $q \gets \Call{ProjectAtConsumer}{e.Q}$
    \Else
      \State $(p,v) \gets \Call{MaterializeSurface}{p,v,e.Q}$
    \EndIf
    \If{$e$ requires a unique authoritative source}
      \State $S \gets e.O$
    \ElsIf{$e$ provides a derived substitution class}
      \State $S \gets$ one representative per class and remote domain
    \Else
      \State $S \gets$ the mapped producer for each demanded consumer
    \EndIf
    \State $p \gets \Call{DeliverByDomain}{p,v,S,e.\cset,e.D}$
    \State $p \gets p + q$
    \State \Return $p$ if $\Call{Check}{p,e}$ accepts; otherwise reject
  \end{algorithmic}
\end{algorithm}

For a token-only edge, the specialized plan projects at the producer and transfers the token. The conservative plan redistributes complete logits and applies the same projection at the consumer. Complete-value and log-probability requests retain layout redistribution and AllGather. A replica plan sends one graph-derived representative to each remote domain and broadcasts locally. Exact shards retain their producer mapping, partials complete the registered reduction, and authoritative decisions use the runtime-selected owner.

\subsection{Reference, Layout-Only Baseline, and Semantic Oracle}

The conservative reference materializes the complete layout, applies any requested projection at the consumer, preserves producer mappings and owner constraints, completes required reductions, and delivers to every demanded consumer. The independent Layout-Only baseline accepts only shape, dtype, source/\allowbreak destination meshes, native domains, and Shard/\allowbreak Replicate/\allowbreak Partial placements. It implements documented layout redistribution without importing \sys's demand, observation, authority, lifetime, checker, or cost model. This baseline measures what layout information alone can establish.

An internal capability ladder attributes demand pruning, domain-local multicast, and replica substitution within \sys. The Semantic Oracle instead shares the full semantic IR, candidate operations, checker, and objective; it enumerates the bounded family and provides a ceiling for the production lowerer. A separate byte-only minimizer greedily reduces remote transfer without consumer semantics and supplies legality counterexamples.

\subsection{Bounded Validation Oracle}

The production lowerer is deterministic; the Semantic Oracle evaluates its Pareto and lexicographic optimality by enumerating a finite family of small topologies. The family contains at most four producers, four demanded consumers, and three consumer domains. Endpoint availability is monotonic; cross-domain operations are single-destination point-to-point transfers; local broadcasts are nonempty and stay within one domain; partials may use an all-producer reduction; and programs contain no cyclic delivery, redundant redelivery, arbitrary backend collective, or overlap schedule.

For each legal plan $p$, the oracle records this footprint:
\begin{equation}
F(p)=\langle B_x,N_x,B_l,N_l,N_r,N_o\rangle ,
\end{equation}
where $B_x$ and $N_x$ are cross-domain bytes and startups, $B_l$ and $N_l$ are local bytes and startups, $N_r$ is the number of reductions, and $N_o$ is total operations. The oracle tests whether the specialized plan lies on the Pareto frontier and minimizes this tuple lexicographically. Within this bounded family, the production lowerer is Pareto efficient and lexicographically optimal.

\subsection{Complexity and Operation Coalescing}

Lowering is a deterministic pass over boundary edges. Within an edge, it groups demanded consumers by native domain and, when substitution is available, by producer class. Operations follow a stable endpoint order, and only an accepted plan produces runtime configuration.

Let $|P|$ and $|C|$ be the number of producer and consumer endpoints for an edge. Contract validation and direct mappings require $O(|P|+|C|)$ state. Grouping consumers and producers uses ordered maps and requires $O((|P|+|C|)\log(|P|+|C|))$ time in the current implementation. The checker executes $m$ emitted operations and tracks endpoint state in $O(|P|+|C|+m)$ space, with bounds independent of tensor payload size because lowering manipulates graph metadata rather than model values. The exhaustive oracle is intentionally separate because its search grows combinatorially and is used only for small-family validation.

Graph-level costing can coalesce payloads that share a transport group, operation kind, source set, destination set, and cross-domain flag. Coalescing affects the estimated startup count but not legality: the checker retains every constituent edge and verifies it against its own consumer observation before the packed operation is costed. This distinction prevents a frame-packing optimization from erasing the provenance or completion requirement of an individual edge.

\subsection{Relation to Collective Synthesis}

Collective synthesis and tuning operate after the logical requirement has been selected. SCCL and TACCL generate topology-specific collective algorithms~\cite{cai2021sccl,shah2023taccl}; AutoCCL tunes low-level parameters of a requested NCCL collective~\cite{xu2025autoccl}; HeteCCL synthesizes schedules for heterogeneous links~\cite{hei2026heteccl}; and MSCCL++ exposes portable primitives and a communication DSL~\cite{hwang2026mscclpp}. \sys can emit one representative transfer plus a domain-local broadcast, after which any suitable backend can optimize those remaining operations. Collective synthesis answers how to execute a requirement; \sys first determines whether and what must cross the boundary.

\section{Cross-Stack Realization}

\subsection{Runtime Architecture}

The prototype connects two independent serving engines rather than forming one heterogeneous native process group. The head partition executes on CUDA/NCCL and the tail partition on CANN/HCCL. Each engine retains its vendor runtime for local tensor-parallel collectives. A cross-domain adapter maps accepted logical plans to runtime switches that control forward aggregation, result aggregation, and node-local result fanout. The current mapping requires the hidden and residual forward edges to share one aggregation decision and requires exactly one authoritative decision-return edge.

The implementation builds on a vLLM-style serving engine~\cite{kwon2023pagedattention}. Cross-domain tensor movement uses persistent binary connections, framed messages, demand-allocated pinned-host staging, depth-two sender queues, and inline receive. Forward values and result decisions follow distinct registered paths. The contract frontend derives the consumer observation from graph and request/runtime APIs. After legality checking, the runtime mapper lowers the accepted plan to the forward and result switches. Requests for log probabilities retain the full-logits redistribution path. The complete runtime topology and message formats are detailed in the Supplemental Material.

\subsection{Data, Result, and Control Paths}

The data path carries the hidden state and residual values produced by the head partition. Under shard forwarding, each local TP lane serializes its mapped value for the corresponding remote lane. Under representative forwarding, one source lane serializes the logical value once, a designated remote lane receives it, and the tail engine invokes its native local collective for the remaining demanded lanes. Hidden and residual payloads share one aggregation decision in the current adapter.

The result path is selected from the consumer observation. The native tail output head materializes full logits before either plan. A complete-logits or log-probability obligation slices that tensor by TP lane, transfers the shards, and reconstructs the full value with a destination-local AllGather. For the deterministic sampling requests used in our evaluation, a token-only obligation applies the native sampler to the same tensor and transfers the compact \texttt{int32} IDs. If the contract establishes a unique owner, the mapper sends one owner value per remote domain and uses a native local broadcast. Projection weakens the observed value; authority selects its admissible source.

Persistent connections amortize connection setup across tokens. Each frame contains a fixed header, tensor metadata, execution tag, and payload. Receive operations are bounded by the declared frame size and expected tensor contract. Pinned staging is allocated on demand and returned to a 512-MiB per-process cache after use; the transport reserves no persistent device-buffer pool. The depth-two sender queue bounds in-flight staging, while inline receive avoids the measured overhead of a persistent receiver thread. These mechanisms are implementation choices for the current backend; the logical plan does not depend on sockets, host staging, or a particular vendor API.

The control path coordinates request admission and retirement. Local ranks first agree on continue, admit, or retire. A host acknowledges cross-node retirement after every local rank drains the current collective slot, and the peer acknowledgement advances the shared epoch. This ordering prevents one rank from entering the next request while another remains in the previous local collective.

\subsection{Plan Realization}

All configurations use the same runtime. Figure~\ref{fig:pipeline-swimlane} profiles local, non-additive result-path timers across the two plans during steady decode. Full-logit transfer (R0) slices tail-materialized logits by TP lane, transfers the shards, reconstructs the value with destination-local AllGather, and then applies argmax. Source-side projection (P0) applies the same argmax before transfer and sends one \texttt{int32} result per lane. For forward values, shard forwarding (L0) reconstructs flatten-defined shards with AllGather, while representative forwarding (S0) transfers one complete value and broadcasts it locally. Owner-token delivery (S1) adds owner-only result transfer and local fanout. Contracts contain no configuration identifiers; the runtime settings are generated from the checked plan.

\begin{figure*}[t]
  \centering
  \includegraphics[width=\textwidth]{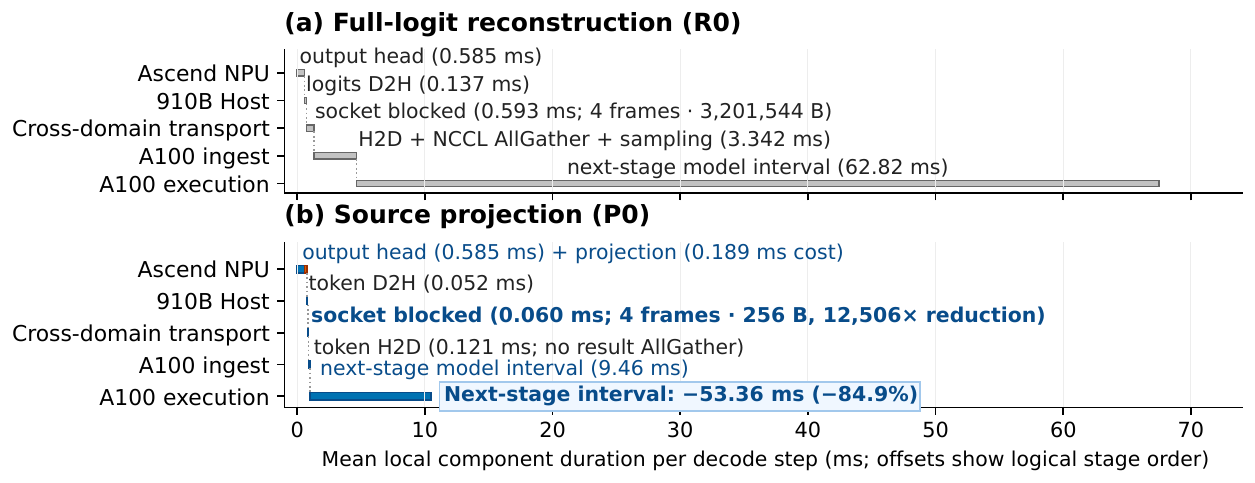}
  \caption{Measured result-path component durations for MiniMax-M2.7 at concurrency 16 on the 1-GbE CUDA/CANN deployment. R0 sends four logit-shard frames totaling 3,201,544\,B per step and performs H2D, NCCL AllGather, and sampling at the A100. P0 adds 0.189\,ms of source projection, sends four token frames totaling 256\,B, and removes result AllGather. Across six restart pairs, the CUDA-stream interval enclosing the subsequent A100 head-stage execution falls from 62.82 to 9.46\,ms. The socket bars measure process-local \texttt{sendall} blocking, not wire-transfer time. Bar widths encode local durations; horizontal offsets indicate logical order; the components are not additive.}
  \Description{Stepped execution swimlane across five hardware lanes comparing full-logit reconstruction with source projection. The bars report separate local timers rather than a summed end-to-end timeline. Source projection collapses cross-domain payload, removes result AllGather, and reduces the measured interval around subsequent A100 head-stage execution from 62.82 to 9.46 milliseconds.}
  \label{fig:pipeline-swimlane}
\end{figure*}

Earlier replica-scaling campaigns use mapped copies (C1), representative broadcast (C0), and representative broadcast with owner return (C5). These configurations retain the original value and frame formats and provide a second view of replica multiplicity. On the TP2 MiniMax graph, they emit forward/result frame pairs of $(2,4)$, $(1,4)$, and $(1,2)$ per logical step. For $k$ graph-derived forward replicas and one remote consumer domain, representative forwarding predicts
\begin{equation}
B_{\mathrm{SemBridge}}\approx\frac{B_{\mathrm{reference}}}{k},
\quad
\mathrm{reduction}\approx 1-\frac{1}{k}.
\label{eq:tp-scaling}
\end{equation}
Equation~\ref{eq:tp-scaling} quantifies cross-domain replica traffic; its end-to-end effect depends on the workload's communication-to-computation ratio.

\subsection{Ordering and Failure Boundary}

Every frame carries request, token, microbatch, and state-epoch fields. Service control uses named continue, admit, and retire states. Before an epoch retires, every local tensor-parallel rank drains its current collective slot; the two hosts then exchange a retirement acknowledgement. These mechanisms prevent a faster rank or host from reusing a collective position while its peer still executes the previous cycle. Transport closure and peer timeouts surface to the service layer as execution errors.

\subsection{Backend Scope and Portability}

The logical IR is backend neutral, while the exercised runtime maps the registered observations onto a CUDA/CANN pair. The serving adapter realizes complete-logits redistribution, token projection, replica forwarding, and owner-sourced decisions. The Qwen3 MLP slice realizes exact shards through lane-preserving point-to-point frames and completes partials with HCCL AllReduce. A backend needs point-to-point delivery, native-domain fanout and reduction, tag validation, projection support, and completion reporting; the logical plan does not depend on socket framing or buffer layout.

\section{Evaluation}
\label{sec:evaluation}

The evaluation follows the abstraction's evidence chain. It first tests whether layout alone establishes consumer observations, then whether checked refinements preserve them on CUDA/CANN engines, and finally whether graph differences affect end-to-end inference.

\subsection{Methodology}

\paragraph{Hardware and software.}
The measured environment contains two physical hosts connected by a dedicated 1-GbE link with a worst-direction p95 TCP-echo RTT of 0.3425\,ms. The CUDA/NCCL host contains four NVIDIA A100 PCIe 40-GB GPUs; the CANN/HCCL host contains four Ascend 910B4-1 accelerators. This topology represents capacity spillover between independently operated accelerator pools over a bandwidth-constrained inter-domain path.

MiniMax-M2.7~\cite{minimax2026m2} is the primary capacity-spillover workload. Qwen3-14B Dense~\cite{qwen2025qwen3} supplies TP scaling, result projection, and the layer-0 MLP slice. Qwen3.5 MoE~\cite{qwen2026qwen35} provides a distinct routing and compute balance under the same compiler rules. Configuration selection precedes the measured runs. For each workload, the protocol fixes output lengths, primary endpoint, AB/BA orders, paired analysis, and inclusion rules before the six measured fresh-process pairs. The MiniMax load sweep holds the token-only API, prompt, and 128-token output fixed while varying concurrency over 1, 8, and 16. All paired configurations use the same checkpoint, requests, process placement, transport, framing, queue depth, and lifecycle within a workload.

\paragraph{Configurations and metrics.}
Full-logit transfer (R0) and source-side projection (P0) satisfy the same token-only request. R0 transfers and reconstructs complete logits before applying the contract argmax; P0 applies the same projection before the boundary. Log-probability requests use full-logit transfer. The forward-path comparison isolates shard transfer plus AllGather (L0) from representative transfer plus Broadcast (S0). Owner-token delivery (S1) adds owner-only result transfer and local fanout. Earlier replica-scaling campaigns use mapped replica copies (C1), representative broadcast (C0), and representative broadcast with owner return (C5). We report output-token throughput, p95 time to first token (TTFT), p95 time per output token (TPOT), p95 total latency, framed bytes, frames, local collectives, and observation checks. Performance estimates use uninstrumented runs; component timers, resource probes, and correctness assertions execute separately.

\paragraph{Statistical units.}
Each AB/BA pair from fresh processes is an independent performance unit; requests within one run are repeated service observations. Every projection-placement and forward/owner comparison contains six pairs, with three per order stratum. For specialized plan $s$ and reference plan $r$, pair $i$ contributes $d_i=100(T_{s,i}/T_{r,i}-1)$. These comparisons report all six effects and the two-sided paired-$t$ interval $\bar d\pm t_{0.975,5}s_d/\sqrt{6}$. The replica-scaling measurements use percentile-bootstrap intervals over six fresh-process pairs. The Supplemental reports every pair effect and both interval estimators for both experiment families. Correctness probes and instrumented runs are excluded from performance estimates.

\begin{figure*}[!t]
  \centering
  \includegraphics[width=\textwidth]{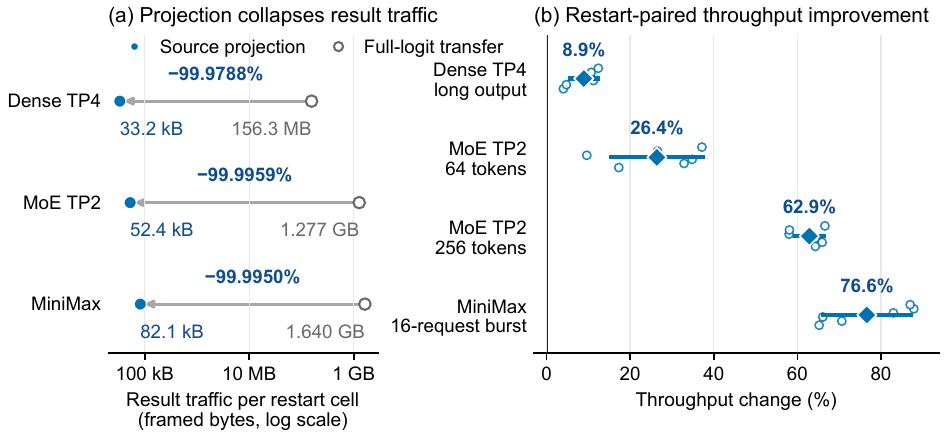}
  \caption{Projection placement in the common runtime. (a) Source-side projection reduces result traffic by more than 99.97\%. (b) Circles show six restart-pair throughput effects; diamonds and bars show means and paired-$t$ 95\% confidence intervals. Forward traffic and result-frame counts match between plans.}
  \Description{Two-panel result plot. Log-scale arrows show result traffic falling from complete logits to token projection for Dense, MoE, and MiniMax. Paired scatter plots and confidence intervals show positive output-token throughput effects in all six restarts for every reported workload.}
  \label{fig:main-results}
\end{figure*}

\subsection{Compiler Comparison with Layout-Only}

The generated matrix contains 72 cases, with 18 each for replica, exact-shard, authority, and partial delivery obligations. It covers TP degrees 1, 2, and 4 and one to three consumer domains. Result projection is orthogonal to these delivery classes and is evaluated separately through compiler checks and the projection-placement experiments. Layout-Only supports every source/destination layout transition and establishes 54 complete application obligations. Table~\ref{tab:compiler} separates footprint reductions on this common semantic subset from the authority information absent from layout.

\begin{table}[t]
  \centering
  \caption{Comparison with the independent Layout-Only baseline across 72 cases. ``Established'' counts complete application obligations derived from layout inputs. Byte and startup columns count cases with lower \sys footprints; ``Rejected'' counts semantically invalid candidates proposed by the byte-only minimizer and rejected by \sys.}
  \label{tab:compiler}
  \footnotesize
  \setlength{\tabcolsep}{2.5pt}
  \begin{tabular}{l r r r r r}
    \toprule
    Obligation & Cases & Established & Bytes $\downarrow$ & Startups $\downarrow$ & Rejected \\
    \midrule
    Replica   & 18 & 18 & 13 & 15 & 0 \\
    Exact shard & 18 & 18 & 7 & 7 & 14 \\
    Partial   & 18 & 18 & 10 & 13 & 18 \\
    \cmidrule(lr){1-6}
    Layout subtotal & 54 & 54 & 30 & 35 & 32 \\
    Authority residual & 18 & 0 & 10 & 13 & 8 \\
    \midrule
    All       & 72 & 54 & 40 & 48 & 40 \\
    \bottomrule
  \end{tabular}
\end{table}

Across the 54 layout-established obligations, \sys uses fewer bytes in 30 cases and fewer startups in 35. Layout alone cannot establish the owner in the 18 authority cases; the checked owner constraint produces lower byte and startup footprints in 10 and 13 cases. On the nine real headline edges, \sys and the layout-only baseline select different plan shapes in every case. \sys reduces startups on all nine and bytes on the three result edges. The shared-IR oracle matches \sys on all 72 generated cases and all 13 real edges, showing that the deterministic lowerer reaches the best footprint in its bounded candidate family. The oracle enumerates 69,322 checker-accepted programs, and 504 legal programs execute against the reference semantics. The byte-only minimizer proposes semantically invalid candidates in 14 exact-shard, 18 partial, and eight authority cases; \sys rejects all 40. All eight targeted mutations of demand, reduction, substitution, ownership, domain locality, and execution identity are rejected.

The common-runtime studies validate these plan differences under live execution. On Dense TP4, shard transfer and representative transfer move the same 14,192,640 payload bytes per fresh-process run. Representative transfer reduces forward frames from 468 to 117 and improves throughput by 1.48\% [0.60\%, 2.36\%] across all six pairs. On Qwen3.5 MoE, the same forward rewrite changes output-256 throughput by $-3.14\%$ [$-5.64\%$, $-0.64\%$]. Owner-token delivery instead improves throughput by 7.98\% [0.53\%, 15.42\%] at output-64 and 6.09\% [2.37\%, 9.81\%] at output-256. Dense benefits from fewer forward startups; MoE benefits from authority-aware result delivery.

Contract compilation, lowering, and checking remain off the token path. Across the model contracts, p95 compilation time is 0.051--0.056\,ms on the A100 host and 0.114--0.130\,ms on the Ascend host.

\subsection{Projection Placement Changes End-to-End Execution}

Full-logit transfer and source-side projection satisfy the same token-only surface, retain identical forward traffic, and use the same number of result frames. The former transfers complete logit shards, invokes result AllGather, and applies argmax at the consumer; the latter applies the same argmax before transfer. Figure~\ref{fig:main-results} places this boundary transformation beside its restart-paired service effect.

The Dense and MoE campaigns complete all 24 fresh-process runs and 288 requests. Source-side projection removes 130 and 332 result AllGathers, respectively, and improves p95 TPOT by 10.46\% on Dense and 41.95\% on the longer MoE output.

The MiniMax load sweep contains six fresh-process pairs at each concurrency. Source-side projection improves throughput by 18.64\% [17.46\%, 19.82\%], 80.20\% [77.12\%, 83.28\%], and 76.63\% [65.61\%, 87.65\%] at concurrency 1, 8, and 16; all 18 pair effects are positive. The corresponding p95 TPOT reductions are 17.27\%, 54.43\%, and 56.24\%, while p95 total latency falls by 15.71\%, 44.54\%, and 43.29\%.

Full-logit runs reconstruct the transferred logits bitwise, and source projection produces identical tokens across all three model families. A live surface-switching run composes result surfaces with delivery obligations through one compiler and CUDA/CANN runtime, without plan identifiers in the contracts. Log-probability requests select full-logit reconstruction (36 frames, 13,609,032\,B) and return all 16 requested positions. Token-only requests select source projection with mapped delivery (20 frames, 1,232\,B), while owner-scoped requests compose the token surface with unique authority and select owner-token delivery (5 frames, 308\,B) with rank~3/lane~0 as the sole sender. Both token plans match 40/40 consumers.

\subsection{Exact Shards and Partial Completion on Real Backends}

The MLP down projection is the canonical tensor-parallel completion boundary: each lane produces an incomplete full-width output that must be summed before a complete downstream observation. The layer-level experiment partitions the gate and up projections of the real Qwen3-14B layer-0 MLP across two A100 GPUs. Each rank produces a distinct 8704-column BF16 activation shard and sends it to the matching Ascend rank through the production framed transport. The two Ascend ranks apply row-parallel down projections, producing distinct partials, then complete the sum with HCCL AllReduce.

\begin{table}[t]
  \centering
  \caption{Qwen3-14B layer-0 MLP correctness. The 8- and 64-token cases were held out from error-gate selection.}
  \label{tab:semantic-slice}
  \footnotesize
  \setlength{\tabcolsep}{2.0pt}
  \begin{tabular}{@{}r c c r r@{}}
    \toprule
    Tokens & \begin{tabular}[c]{@{}c@{}}Shard\\delivery\end{tabular} & \begin{tabular}[c]{@{}c@{}}Partial\\completion\end{tabular} & \begin{tabular}[c]{@{}c@{}}Relative\\L2 (\%)\end{tabular} & \begin{tabular}[c]{@{}c@{}}Max error\\/ peak (\%)\end{tabular} \\
    \midrule
    1  & bitwise & distinct $\rightarrow$ equal & 0.207 & 0.164 \\
    8  & bitwise & distinct $\rightarrow$ equal & 0.226 & 0.324 \\
    64 & bitwise & distinct $\rightarrow$ equal & 0.244 & 0.451 \\
    \bottomrule
  \end{tabular}
\end{table}

For every token count, each received activation matches its sender bitwise, the two activation lanes differ, the pre-reduction partials differ, and the two completed outputs are identical. Relative L2 error against a full FP32 down-projection reference is 0.207--0.244\%, establishing live exact-shard delivery and partial completion on the two vendor stacks.

\subsection{Legal Plans Have Workload-Dependent Effects}

Figure~\ref{fig:replica-scaling} tests Equation~\ref{eq:tp-scaling} without changing the result surface. Representative broadcast has no replica advantage at TP1, reduces host traffic by 50.6--52.5\% at TP2, and by 75.2--76.2\% at TP4. The MiniMax variants with owner return realize forward/result frame pairs $(2,4)$, $(1,4)$, and $(1,2)$ per step.

\begin{figure}[t]
  \centering
  \includegraphics[width=\columnwidth]{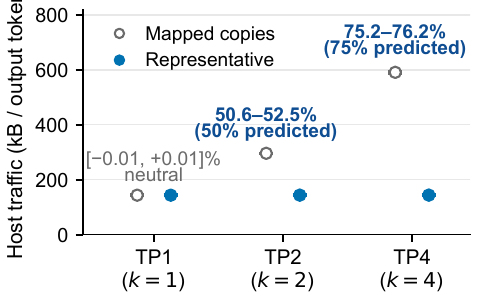}
  \caption{Replica-aware traffic scaling. Host traffic per output token for mapped copies and representative broadcast; ranges cover directions and restart blocks.}
  \Description{Mapped-copy traffic increases with tensor-parallel degree, while representative-broadcast traffic remains nearly constant from TP1 through TP4.}
  \label{fig:replica-scaling}
\end{figure}

Across six independent fresh-process pairs, representative broadcast raises throughput by 38.57\% [36.89\%, 40.07\%] on Dense TP4 long output, 13.68\% [10.15\%, 17.82\%] on a 16-request MiniMax burst, and 19.17\% [13.20\%, 25.23\%] on MoE with long output at concurrency 8. The steady-low MiniMax control is practically neutral at $+0.54\%$, Dense TP1 long output averages $-1.87\%$ without replica multiplicity, and the 16-request MoE burst averages $-10.84\%$ with an interval crossing zero. Together, these measurements show how \sys exposes semantically valid alternatives whose service benefit follows the workload critical path.

\paragraph{Resource footprint.}
P0 leaves peak allocated and reserved device memory, KV blocks, maximum batch size, and concurrency unchanged from R0. It lowers pinned-memory high-water by 37.40\,MiB on the A100 host and 74.79\,MiB on the Ascend host. Neither plan allocates a persistent transport device-buffer pool, and the sender queue high-water remains unchanged.

\section{Related Work}

\subsection{Placement and Distributed Tensors}

Distributed-tensor systems separate a global program from device-level realization. Mesh-TensorFlow, GShard, GSPMD, DistIR, and PartIR express or propagate placement and sharding decisions~\cite{shazeer2018mesh,lepikhin2021gshard,xu2021gspmd,santhanam2021distir,alabed2025partir}. OneFlow SBP and PyTorch DTensor represent split or shard, broadcast or replicate, and partial states~\cite{yuan2021oneflow,pytorch2026dtensor}, while cross-mesh resharding optimizes layout transfer~\cite{zhuang2023optimizing}. \sys starts after placement: its contract determines which observations must cross independent vendor domains.

TrainVerify verifies parallelization equivalence between a logical model and a distributed training plan through symbolic dataflow graphs and solver-backed reasoning; Scalify checks production graph transformations through equality saturation and relational layout analysis~\cite{lu2025trainverify,zulkifli2025verifying}. These systems establish whole-graph equivalence. \sys checks consumer observations at a mixed-vendor communication boundary and lowers each certified obligation into a communication graph.

\subsection{Heterogeneous LLM Serving}

PagedAttention and Sarathi-Serve improve memory management and scheduling~\cite{kwon2023pagedattention,agrawal2024sarathi}; Splitwise and DistServe separate prefill from decode~\cite{patel2024splitwise,zhong2024distserve}; and HexGen, HexGen-2, and Helix optimize heterogeneous placement or routing~\cite{jiang2024hexgen,jiang2025hexgen2,mei2025helix}. These systems decide where computation and requests run. \sys instead asks which boundary values and copies the fixed placement requires.

Production pipeline runtimes also communicate compact sampler outputs in selected homogeneous paths. The vLLM GPU runner broadcasts sampled token IDs across pipeline ranks, and the vLLM-Ascend sampling design gathers compact sampler outputs rather than complete logits~\cite{vllm2026ppoutput,vllmascend2026sampling}. A path-specific token handoff realizes the same P0 graph once eligibility is known. \sys derives that eligibility from the requested surface, retains R0 for log-probability requests, and composes the token result surface with completion and authority obligations across separate CUDA/NCCL and CANN/HCCL domains.

\subsection{Collective Synthesis and Mixed-Vendor Communication}

Blink, SCCL, P$^2$, TACCL, and MSCCLang construct or synthesize collective programs~\cite{wang2020blink,cai2021sccl,xie2022reduction,shah2023taccl,cowan2023mscclang}. AutoCCL tunes low-level parameters~\cite{xu2025autoccl}; HeteCCL handles heterogeneous link capacities~\cite{hei2026heteccl}; and MSCCL++ and HetCCL provide portable or mixed-vendor realization~\cite{hwang2026mscclpp,wang2026hetccl}. These systems implement or optimize a logical requirement. \sys selects the consumer-observable requirement that reaches them.

\section{Discussion}

\subsection{Compiler Boundary}

Placement fixes where values reside, \sys specifies what each consumer must observe, and collective backends determine how the checked logical graph executes. The contract makes the result surface explicit before transport: token-only requests produce a projection operation, log-probability requests retain complete logits, replica lineage enables representative transfer, and decision ownership selects the source. The checked plan then maps to ordinary transfers and native collectives. Full-logit transfer and source-side projection use identical forward traffic and result-frame counts, so their end-to-end difference comes from this compiler decision within one transport implementation.

\subsection{Performance and Fabric-Aware Selection}

In the measured capacity-spillover deployment, projection reduces result traffic by 99.9788--99.9959\%, raises throughput by 8.92--80.20\%, and lowers p95 total latency by 8.12--44.54\% across the evaluated model and load points. Every MiniMax restart pair favors source projection at concurrency 1, 8, and 16. Representative delivery separately reduces TP4 traffic by 75.2--76.2\% and raises throughput by 13.68--38.57\% on communication-sensitive workloads.

Compilation first uses the contract and checker to form the set of legal plans; a fabric-aware selector can then choose among them. Let $B_i,F_i$ be plan $i$'s cross-domain payload bytes and frames, $\beta$ the effective end-to-end payload rate, $\alpha$ the effective non-overlapped per-frame startup, and $L_i$ the plan's non-overlapped projection and local-collective cost. The effective payload rate includes device-to-host serialization, framing, network transfer, and host-to-device ingestion; local reconstruction and projection contribute to $L_i$. A specialized plan $s$ improves over reference $r$ when
\begin{equation}
  \frac{B_r-B_s}{\beta}+\alpha(F_r-F_s)>L_s-L_r.
  \label{eq:crossover}
\end{equation}
Byte and frame reductions strengthen the left-hand side; projection and native collectives determine the local-cost difference. The Dense and MoE forward/owner results show that two legal rewrites can occupy different critical paths. Restart pairs measure the net outcome, while the non-additive component timers identify transfer and local work. Equation~\ref{eq:crossover} provides a decision criterion for checker-accepted plans: transfer savings must exceed projection and local-collective cost.

\section{Conclusion}

SemBridge makes consumer observations explicit at cross-stack communication boundaries through a typed contract, deterministic lowerer, and symbolic checker. Its compiler/runtime path derives and executes full-logit reconstruction, source projection, owner-token delivery, representative delivery, exact-shard transfer, and partial completion. Across three model families, token projection reduces result traffic by more than 99.97\% and improves throughput throughout the measured MiniMax load range.

\section*{Acknowledgments}
Generative AI tools were used to assist with drafting and revising portions of the manuscript text. The authors reviewed and verified all generated content and take full responsibility for the final manuscript.

\bibliographystyle{ACM-Reference-Format}
\bibliography{references}

@inproceedings{alabed2025partir,
  author    = {Sami Alabed and Daniel Belov and Bart Chrzaszcz and Juliana Franco and Dominik Grewe and Dougal Maclaurin and James Molloy and Tom Natan and Tamara Norman and Xiaoyue Pan and Adam Paszke and Norman A. Rink and Michael Schaarschmidt and Timur Sitdikov and Agnieszka Swietlik and Dimitrios Vytiniotis and Joel Wee},
  title     = {{PartIR}: Composing {SPMD} Partitioning Strategies for Machine Learning},
  booktitle = {Proceedings of the 30th ACM International Conference on Architectural Support for Programming Languages and Operating Systems, Volume 1},
  pages     = {794--810},
  year      = {2025},
  publisher = {ACM},
  doi       = {10.1145/3669940.3707284},
  url       = {https://doi.org/10.1145/3669940.3707284}
}

@inproceedings{zhuang2023optimizing,
  author    = {Yonghao Zhuang and Lianmin Zheng and Zhuohan Li and Eric P. Xing and Qirong Ho and Joseph E. Gonzalez and Ion Stoica and Hao Zhang and Hexu Zhao},
  title     = {On Optimizing the Communication of Model Parallelism},
  booktitle = {Proceedings of Machine Learning and Systems},
  pages     = {524--540},
  volume    = {5},
  year      = {2023},
  publisher = {Curran Associates},
  url       = {https://proceedings.mlsys.org/paper_files/paper/2023/hash/a42cbafcabb6dc7ce77bfe2e80f5c772-Abstract-mlsys2023.html}
}

@inproceedings{zheng2022alpa,
  author    = {Lianmin Zheng and Zhuohan Li and Hao Zhang and Yonghao Zhuang and Zhifeng Chen and Yanping Huang and Yida Wang and Yuanzhong Xu and Danyang Zhuo and Eric P. Xing and Joseph E. Gonzalez and Ion Stoica},
  title     = {Alpa: Automating Inter- and {Intra-Operator} Parallelism for Distributed Deep Learning},
  booktitle = {16th USENIX Symposium on Operating Systems Design and Implementation (OSDI 22)},
  pages     = {559--578},
  year      = {2022},
  publisher = {USENIX Association},
  isbn      = {978-1-939133-28-1},
  url       = {https://www.usenix.org/conference/osdi22/presentation/zheng-lianmin}
}

@article{xu2021gspmd,
  author  = {Yuanzhong Xu and HyoukJoong Lee and Dehao Chen and Blake Hechtman and Yanping Huang and Rahul Joshi and Maxim Krikun and Dmitry Lepikhin and Andy Ly and Marcello Maggioni and Ruoming Pang and Noam Shazeer and Shibo Wang and Tao Wang and Yonghui Wu and Zhifeng Chen},
  title   = {{GSPMD}: General and Scalable Parallelization for {ML} Computation Graphs},
  journal = {CoRR},
  volume  = {abs/2105.04663},
  year    = {2021},
  doi     = {10.48550/arXiv.2105.04663},
  url     = {https://arxiv.org/abs/2105.04663}
}

@inproceedings{shazeer2018mesh,
  author    = {Noam Shazeer and Youlong Cheng and Niki Parmar and Dustin Tran and Ashish Vaswani and Penporn Koanantakool and Peter Hawkins and HyoukJoong Lee and Mingsheng Hong and Cliff Young and Ryan Sepassi and Blake Hechtman},
  title     = {{Mesh-TensorFlow}: Deep Learning for Supercomputers},
  booktitle = {Advances in Neural Information Processing Systems},
  volume    = {31},
  pages     = {10435--10444},
  year      = {2018},
  url       = {https://proceedings.neurips.cc/paper/2018/hash/3a37abdeefe1dab1b30f7c5c7e581b93-Abstract.html}
}

@inproceedings{jia2019flexflow,
  author    = {Zhihao Jia and Matei Zaharia and Alex Aiken},
  title     = {Beyond Data and Model Parallelism for Deep Neural Networks},
  booktitle = {Proceedings of Machine Learning and Systems},
  volume    = {1},
  pages     = {1--13},
  year      = {2019},
  url       = {https://proceedings.mlsys.org/paper_files/paper/2019/hash/b422680f3db0986ddd7f8f126baaf0fa-Abstract.html}
}

@inproceedings{santhanam2021distir,
  author    = {Keshav Santhanam and Siddharth Krishna and Ryota Tomioka and Andrew Fitzgibbon and Tim Harris},
  title     = {{DistIR}: An Intermediate Representation for Optimizing Distributed Neural Networks},
  booktitle = {Proceedings of the 1st Workshop on Machine Learning and Systems},
  pages     = {15--23},
  year      = {2021},
  publisher = {ACM},
  doi       = {10.1145/3437984.3458829},
  url       = {https://doi.org/10.1145/3437984.3458829}
}

@inproceedings{jia2022whale,
  author    = {Xianyan Jia and Le Jiang and Ang Wang and Wencong Xiao and Ziji Shi and Jie Zhang and Xinyuan Li and Langshi Chen and Yong Li and Zhen Zheng and Xiaoyong Liu and Wei Lin},
  title     = {Whale: Efficient Giant Model Training over Heterogeneous {GPUs}},
  booktitle = {2022 USENIX Annual Technical Conference (USENIX ATC 22)},
  pages     = {673--688},
  year      = {2022},
  publisher = {USENIX Association},
  isbn      = {978-1-939133-29-57},
  url       = {https://www.usenix.org/conference/atc22/presentation/jia-xianyan}
}

@inproceedings{mei2025helix,
  author    = {Yixuan Mei and Yonghao Zhuang and Xupeng Miao and Juncheng Yang and Zhihao Jia and Rashmi Vinayak},
  title     = {Helix: Serving Large Language Models over Heterogeneous {GPUs} and Network via Max-Flow},
  booktitle = {Proceedings of the 30th ACM International Conference on Architectural Support for Programming Languages and Operating Systems, Volume 1},
  pages     = {586--602},
  year      = {2025},
  publisher = {ACM},
  doi       = {10.1145/3669940.3707215},
  url       = {https://doi.org/10.1145/3669940.3707215}
}

@inproceedings{jiang2024hexgen,
  author    = {Youhe Jiang and Ran Yan and Xiaozhe Yao and Yang Zhou and Beidi Chen and Binhang Yuan},
  title     = {{HexGen}: Generative Inference of Large Language Model over Heterogeneous Environment},
  booktitle = {Proceedings of the 41st International Conference on Machine Learning},
  series    = {Proceedings of Machine Learning Research},
  volume    = {235},
  pages     = {21946--21961},
  year      = {2024},
  publisher = {PMLR},
  url       = {https://proceedings.mlr.press/v235/jiang24f.html}
}

@inproceedings{jiang2025hexgen2,
  author    = {Youhe Jiang and Ran Yan and Binhang Yuan},
  title     = {{HexGen-2}: Disaggregated Generative Inference of {LLMs} in Heterogeneous Environment},
  booktitle = {International Conference on Learning Representations},
  year      = {2025},
  url       = {https://proceedings.iclr.cc/paper_files/paper/2025/hash/0b941a1e5fbce23fe46b049999d04ed0-Abstract-Conference.html}
}

@inproceedings{shah2023taccl,
  author    = {Aashaka Shah and Vijay Chidambaram and Meghan Cowan and Saeed Maleki and Madan Musuvathi and Todd Mytkowicz and Jacob Nelson and Olli Saarikivi and Rachee Singh},
  title     = {{TACCL}: Guiding Collective Algorithm Synthesis using Communication Sketches},
  booktitle = {20th USENIX Symposium on Networked Systems Design and Implementation (NSDI 23)},
  pages     = {593--612},
  year      = {2023},
  publisher = {USENIX Association},
  isbn      = {978-1-939133-33-5},
  url       = {https://www.usenix.org/conference/nsdi23/presentation/shah}
}

@inproceedings{xu2025autoccl,
  author    = {Guanbin Xu and Zhihao Le and Yinhe Chen and Zhiqi Lin and Zewen Jin and Youshan Miao and Cheng Li},
  title     = {{AutoCCL}: Automated Collective Communication Tuning for Accelerating Distributed and Parallel {DNN} Training},
  booktitle = {22nd USENIX Symposium on Networked Systems Design and Implementation (NSDI 25)},
  pages     = {667--683},
  year      = {2025},
  publisher = {USENIX Association},
  isbn      = {978-1-939133-46-5},
  url       = {https://www.usenix.org/conference/nsdi25/presentation/xu-guanbin}
}

@inproceedings{hei2026heteccl,
  author    = {Chenyang Hei and Fuliang Li and Jiayi Li and Jiamin Cao and Chengxi Gao and Xiuzhu Sha and Tongrui Liu and Dengke Zhang and Ennan Zhai and Xingwei Wang},
  title     = {{HeteCCL}: Synthesizing Near-Optimal Collective Communication Schedules for Heterogeneous {GPU} Clusters},
  booktitle = {23rd USENIX Symposium on Networked Systems Design and Implementation (NSDI 26)},
  pages     = {2533--2551},
  year      = {2026},
  publisher = {USENIX Association},
  isbn      = {978-1-939133-54-0},
  url       = {https://www.usenix.org/conference/nsdi26/presentation/hei}
}

@article{wang2026hetccl,
  author  = {Yuejie Wang and Tao Chang and Yuanyuan Zhao and Yulong Ao and Zeyu Gu and Zhiyu Li and Yanmin Jia and Yan Zhang and Mingjun Zhang and He Liu and Yongzhe He and Yonghua Lin and Guyue Liu},
  title   = {{HetCCL}: Enabling Collective Communication for Mixed-Vendor Heterogeneous Clusters},
  journal = {CoRR},
  volume  = {abs/2605.31000},
  year    = {2026},
  doi     = {10.48550/arXiv.2605.31000},
  url     = {https://arxiv.org/abs/2605.31000}
}

@inproceedings{cai2021sccl,
  author    = {Zixian Cai and Zhengyang Liu and Saeed Maleki and Madanlal Musuvathi and Todd Mytkowicz and Jacob Nelson and Olli Saarikivi},
  title     = {Synthesizing Optimal Collective Algorithms},
  booktitle = {Proceedings of the 26th ACM SIGPLAN Symposium on Principles and Practice of Parallel Programming},
  pages     = {62--75},
  year      = {2021},
  publisher = {ACM},
  doi       = {10.1145/3437801.3441620},
  url       = {https://doi.org/10.1145/3437801.3441620}
}

@inproceedings{wang2020blink,
  author    = {Guanhua Wang and Shivaram Venkataraman and Amar Phanishayee and Nikhil Devanur and Jorgen Thelin and Ion Stoica},
  title     = {Blink: Fast and Generic Collectives for Distributed {ML}},
  booktitle = {Proceedings of Machine Learning and Systems},
  volume    = {2},
  pages     = {172--186},
  year      = {2020},
  url       = {https://proceedings.mlsys.org/paper_files/paper/2020/hash/cd3a9a55f7f3723133fa4a13628cdf03-Abstract.html}
}

@inproceedings{hwang2026mscclpp,
  author    = {Changho Hwang and Peng Cheng and Roshan Dathathri and Abhinav Jangda and Saeed Maleki and Madan Musuvathi and Olli Saarikivi and Aashaka Shah and Ziyue Yang and Binyang Li and Caio Rocha and Qinghua Zhou and Mahdieh Ghazimirsaeed and Sreevatsa Anantharamu and Jithin Jose},
  title     = {{MSCCL++}: Rethinking {GPU} Communication Abstractions for {AI} Inference},
  booktitle = {Proceedings of the 31st ACM International Conference on Architectural Support for Programming Languages and Operating Systems, Volume 2},
  pages     = {1201--1215},
  year      = {2026},
  publisher = {ACM},
  doi       = {10.1145/3779212.3790188},
  url       = {https://doi.org/10.1145/3779212.3790188}
}

@inproceedings{kwon2023pagedattention,
  author    = {Woosuk Kwon and Zhuohan Li and Siyuan Zhuang and Ying Sheng and Lianmin Zheng and Cody Hao Yu and Joseph E. Gonzalez and Hao Zhang and Ion Stoica},
  title     = {Efficient Memory Management for Large Language Model Serving with {PagedAttention}},
  booktitle = {Proceedings of the 29th Symposium on Operating Systems Principles},
  pages     = {611--626},
  year      = {2023},
  publisher = {ACM},
  doi       = {10.1145/3600006.3613165},
  url       = {https://doi.org/10.1145/3600006.3613165}
}

@inproceedings{patel2024splitwise,
  author    = {Pratyush Patel and Esha Choukse and Chaojie Zhang and Aashaka Shah and {\'I\~nigo} Goiri and Saeed Maleki and Ricardo Bianchini},
  title     = {Splitwise: Efficient Generative {LLM} Inference Using Phase Splitting},
  booktitle = {2024 ACM/IEEE 51st Annual International Symposium on Computer Architecture},
  pages     = {118--132},
  year      = {2024},
  publisher = {IEEE},
  doi       = {10.1109/ISCA59077.2024.00019},
  url       = {https://doi.org/10.1109/ISCA59077.2024.00019}
}

@inproceedings{zhong2024distserve,
  author    = {Yinmin Zhong and Shengyu Liu and Junda Chen and Jianbo Hu and Yibo Zhu and Xuanzhe Liu and Xin Jin and Hao Zhang},
  title     = {{DistServe}: Disaggregating Prefill and Decoding for Goodput-Optimized Large Language Model Serving},
  booktitle = {18th USENIX Symposium on Operating Systems Design and Implementation (OSDI 24)},
  pages     = {193--210},
  year      = {2024},
  publisher = {USENIX Association},
  isbn      = {978-1-939133-40-3},
  url       = {https://www.usenix.org/conference/osdi24/presentation/zhong-yinmin}
}

@inproceedings{agrawal2024sarathi,
  author    = {Amey Agrawal and Nitin Kedia and Ashish Panwar and Jayashree Mohan and Nipun Kwatra and Bhargav Gulavani and Alexey Tumanov and Ramachandran Ramjee},
  title     = {Taming Throughput-Latency Tradeoff in {LLM} Inference with {Sarathi-Serve}},
  booktitle = {18th USENIX Symposium on Operating Systems Design and Implementation (OSDI 24)},
  pages     = {117--134},
  year      = {2024},
  publisher = {USENIX Association},
  isbn      = {978-1-939133-40-3},
  url       = {https://www.usenix.org/conference/osdi24/presentation/agrawal}
}

@inproceedings{lepikhin2021gshard,
  author    = {Dmitry Lepikhin and HyoukJoong Lee and Yuanzhong Xu and Dehao Chen and Orhan Firat and Yanping Huang and Maxim Krikun and Noam Shazeer and Zhifeng Chen},
  title     = {{GShard}: Scaling Giant Models with Conditional Computation and Automatic Sharding},
  booktitle = {International Conference on Learning Representations},
  year      = {2021},
  url       = {https://openreview.net/forum?id=qrwe7XHTmYb}
}

@misc{minimax2026m2,
  author       = {{MiniMax}},
  title        = {{MiniMax M2.7}: Early Echoes of Self-Evolution},
  year         = {2026},
  howpublished = {Official model release},
  url          = {https://www.minimax.io/news/minimax-m27-en}
}

@misc{qwen2025qwen3,
  author       = {{Qwen Team}},
  title        = {{Qwen3}: Think Deeper, Act Faster},
  year         = {2025},
  howpublished = {Official model release},
  url          = {https://qwenlm.github.io/blog/qwen3/}
}

@misc{qwen2026qwen35,
  author       = {{Qwen Team}},
  title        = {{Qwen3.5}: Towards Native Multimodal Agents},
  year         = {2026},
  howpublished = {Official model release},
  url          = {https://qwen.ai/blog?id=qwen3.5}
}

@article{yuan2021oneflow,
  author  = {Jinhui Yuan and Xinqi Li and Cheng Cheng and Juncheng Liu and Ran Guo and Shenghang Cai and Chi Yao and Fei Yang and Xiaodong Yi and Chuan Wu and Haoran Zhang and Jie Zhao},
  title   = {{OneFlow}: Redesign the Distributed Deep Learning Framework from Scratch},
  journal = {CoRR},
  volume  = {abs/2110.15032},
  year    = {2021},
  url     = {https://arxiv.org/abs/2110.15032}
}

@misc{pytorch2026dtensor,
  author       = {{PyTorch Contributors}},
  title        = {torch.distributed.tensor: Distributed Tensor Documentation},
  year         = {2026},
  howpublished = {PyTorch documentation},
  url          = {https://docs.pytorch.org/docs/stable/distributed.tensor.html}
}

@inproceedings{jangda2022coconet,
  author    = {Abhinav Jangda and Jun Huang and Guodong Liu and Amir Hossein Nodehi Sabet and Saeed Maleki and Youshan Miao and Madanlal Musuvathi and Todd Mytkowicz and Olli Saarikivi},
  title     = {Breaking the Computation and Communication Abstraction Barrier in Distributed Machine Learning Workloads},
  booktitle = {Proceedings of the 27th ACM International Conference on Architectural Support for Programming Languages and Operating Systems},
  pages     = {402--416},
  publisher = {ACM},
  year      = {2022},
  doi       = {10.1145/3503222.3507778},
  url       = {https://doi.org/10.1145/3503222.3507778}
}

@inproceedings{unger2022unity,
  author    = {Colin Unger and Zhihao Jia and Wei Wu and Sina Lin and Mandeep Baines and Carlos Efrain Quintero Narvaez and Vinay Ramakrishnaiah and Nirmal Prajapati and Pat McCormick and Jamaludin Mohd-Yusof and Xi Luo and Dheevatsa Mudigere and Jongsoo Park and Misha Smelyanskiy and Alex Aiken},
  title     = {Unity: Accelerating {DNN} Training Through Joint Optimization of Algebraic Transformations and Parallelization},
  booktitle = {16th USENIX Symposium on Operating Systems Design and Implementation (OSDI 22)},
  pages     = {267--284},
  publisher = {USENIX Association},
  year      = {2022},
  isbn      = {978-1-939133-28-1},
  url       = {https://www.usenix.org/conference/osdi22/presentation/unger}
}

@inproceedings{lu2025trainverify,
  author    = {Yunchi Lu and Youshan Miao and Cheng Tan and Peng Huang and Yi Zhu and Xian Zhang and Fan Yang},
  title     = {{TrainVerify}: Equivalence-Based Verification for Distributed {LLM} Training},
  booktitle = {Proceedings of the ACM SIGOPS 31st Symposium on Operating Systems Principles},
  pages     = {237--253},
  publisher = {ACM},
  year      = {2025},
  doi       = {10.1145/3731569.3764850},
  url       = {https://doi.org/10.1145/3731569.3764850}
}

@article{zulkifli2025verifying,
  author  = {Kahfi S. Zulkifli and Wenbo Qian and Shaowei Zhu and Yuan Zhou and Zhen Zhang and Chang Lou},
  title   = {Verifying Computational Graphs in Production-Grade Distributed Machine Learning Frameworks},
  journal = {CoRR},
  volume  = {abs/2509.10694},
  year    = {2025},
  doi     = {10.48550/arXiv.2509.10694},
  url     = {https://arxiv.org/abs/2509.10694}
}

@inproceedings{xie2022reduction,
  author    = {Ningning Xie and Tamara Norman and Dominik Grewe and Dimitrios Vytiniotis},
  title     = {Synthesizing Optimal Parallelism Placement and Reduction Strategies on Hierarchical Systems for Deep Learning},
  booktitle = {Proceedings of Machine Learning and Systems},
  volume    = {4},
  publisher = {mlsys.org},
  year      = {2022},
  url       = {https://proceedings.mlsys.org/paper_files/paper/2022/hash/f0f9e98bc2e2f0abc3e315eaa0d808fc-Abstract.html}
}

@inproceedings{cowan2023mscclang,
  author    = {Meghan Cowan and Saeed Maleki and Madanlal Musuvathi and Olli Saarikivi and Yifan Xiong},
  title     = {{MSCCLang}: Microsoft Collective Communication Language},
  booktitle = {Proceedings of the 28th ACM International Conference on Architectural Support for Programming Languages and Operating Systems, Volume 2},
  pages     = {502--514},
  publisher = {ACM},
  year      = {2023},
  doi       = {10.1145/3575693.3575724},
  url       = {https://doi.org/10.1145/3575693.3575724}
}

@misc{vllm2026ppoutput,
  author       = {{vLLM Contributors}},
  title        = {Pipeline-Parallel Sampled-Token Handoff in the {vLLM} {GPU} Model Runner},
  year         = {2026},
  howpublished = {vLLM v0.19.0 source code},
  url          = {https://github.com/vllm-project/vllm/blob/v0.19.0/vllm/v1/worker/gpu_model_runner.py}
}

@misc{vllmascend2026sampling,
  author       = {{vLLM Ascend Contributors}},
  title        = {Sampling Path Optimization for {NPU} Model Runner {V1}},
  year         = {2026},
  howpublished = {vLLM Ascend request for comments, issue 9269},
  url          = {https://github.com/vllm-project/vllm-ascend/issues/9269}
}

\end{document}